\documentclass[a4paper]{article}
\usepackage[margin=1.25in]{geometry}
\usepackage{authblk}
\usepackage[utf8]{inputenc}
\usepackage{microtype}
\usepackage{amsfonts}
\usepackage{amssymb}
\usepackage{amsmath}
\usepackage{amsthm}
\usepackage{thm-restate}
\usepackage{comment}
\usepackage{xspace}
\usepackage{enumitem}
\usepackage[hidelinks]{hyperref}
\usepackage[capitalize]{cleveref}
\usepackage[ruled,linesnumbered,noend]{algorithm2e}

\newtheorem{theorem}{Theorem}
\newtheorem{corollary}[theorem]{Corollary}
\newtheorem{lemma}[theorem]{Lemma}

\newtheorem{invariant}[theorem]{Invariant}

\renewcommand{\split}{\texttt{Split}}

\newcommand{\fail}{\texttt{Fail}}
\newcommand{\lefty}{\texttt{Left}}
\newcommand{\done}{\texttt{Done}}

\newcommand{\sync}{\texttt{Sync}}

\title{An {$O(\log^{3/2}n)$}{} Parallel Time Population Protocol for Majority with {$O(\log n)$}{} States\thanks{Supported by ISF grants no. 1278/16 and 1926/19, by a BSF grant 2018364, and by an ERC grant MPM under the EU's Horizon 2020 Research and Innovation Programme (grant no. 683064).}}

\begin{document}

\renewcommand\Affilfont{\normalsize}

\author[]{Stav Ben-Nun}
\author[]{Tsvi Kopelowitz}
\author[]{Matan Kraus}
\author[]{Ely Porat}

\affil[]{Department of Computer Science, Bar-Ilan University, Ramat Gan, Israel}
\affil[]{\texttt{stav.bennun@live.biu.ac.il, kopelot@gmail.com, krausma@biu.ac.il, porately@cs.biu.ac.il}}

\date{\vspace{-1.5cm}}
\maketitle
	
	\begin{abstract}
		In population protocols,  the underlying distributed network consists of $n$ nodes (or agents), denoted by $V$, and a \emph{scheduler} that continuously selects uniformly random pairs of nodes to \emph{interact}.
		When two nodes interact, their states are updated by applying a state transition function that depends only on the states of the two nodes prior to the interaction.
		The efficiency of a population protocol is measured in terms of both time (which is the number of interactions until the nodes collectively have a valid output) and the number of possible states of nodes used by the protocol.
		By convention, we consider the parallel time cost, which is the time divided by $n$.

		In this paper we consider the \emph{majority} problem, where each node receives as input a color that is either black or white, and the goal is to have all of the nodes output the color that is the majority of the input colors.
		We design a population protocol that solves the majority problem in $O(\log^{3/2}n)$ parallel time, both with high probability and in expectation, while using $O(\log n)$ states. Our protocol improves on a recent protocol of Berenbrink et al. that runs in $O(\log^{5/3}n)$ parallel time, both with high probability and in expectation, using $O(\log n)$ states.

	\end{abstract}

	\section{Introduction}

	Population protocols, introduced by Angluin et al.~\cite{AADF2004}, have been extensively studied in recent years~\cite{AAE2008,AAEGR2017,AAG18,AG2015,AGV2015,AR2009,BCER2017,BEFKKR2018,DV2010,KU2018,MNRS2014,DBLP:conf/dna/AlistarhDKSU17,DBLP:conf/spaa/GasieniecSU19,DBLP:conf/icalp/CzyzowiczGKKSU15}. In these protocols, the underlying distributed network consists of $n$ nodes (or agents), denoted by $V $, and a \emph{scheduler} that continuously selects pairs of nodes to \emph{interact}.
	Each node $u$ stores its own local state $s_u\in S$, and when two nodes $u,v\in V$ interact, the states of $u$ and $v$ after the interaction are determined solely by the states of $u$ and $v$ prior to the interaction.
	Following previous work, we assume that any two nodes are allowed to interact, and that each selection of a pair of nodes for interaction is uniformly random from the pairs of nodes in $V$.
	The objective of population protocols for a given problem is to extract an output from $s_u$ for each $u\in V$ so that the collective configuration of outputs defines a \emph{valid} solution to the problem, and the valid solution will never change, regardless of future interactions. When the system reaches a valid output configuration that never changes, the system is said to be \emph{stable}.
	The number of interactions until the system is stable is called the \emph{stabilization} time of the protocol.
	By convention, we consider the parallel stabilization time, which is the stabilization time divided by $n$.

	The efficiency of a population protocol is measured in terms of both the parallel stabilization time, either with high probability or in expectation, and $|S|$, which directly affects the number of bits that each node must invest in the protocol.
	Thus, if a protocol runs for $f(n)$ parallel time and has $|S|= g(n)$, then we define the efficiency of the protocol by \textless $f(n),g(n)$\textgreater.
	Typically, one is interested in protocols where $|S| = O(\text{poly-}\log(n))$, in which case each node stores only $O(\log \log n)$ bits.

	\paragraph{Majority.}
	In this paper we consider the \emph{majority} problem in population protocols, where each node $u$ receives an input color from $\{b,w\}$, where $b$ should be interpreted as the color \emph{black} and $w$ should be interpreted as the color \emph{white}.
	The output of a node is also a color from $\{b,w\}$.
	An output configuration is valid if all of the nodes output the same color, and more than half of the nodes have this color as their input (the input is assumed to have a well defined majority).

	\paragraph{Previous work.}
	Mertzios et al.~\cite{MNRS2014} and Draief and Vojnović~\cite{DV2010} designed \textless$O(n\log n), O(1)$\textgreater\space protocols for majority. Mertzios et al.~\cite{MNRS2014} named their protocol the \emph{ambassador} protocol, and proved that the parallel stabilization time is $O(n\log n)$.
	Alistarh et al.~\cite{AAEGR2017} recently refined ideas from Alistrah et al.~\cite{AGV2015}, in order to design an \textless$O(\log^3 n), O(\log^2 n)$\textgreater\space majority protocol.
	
	Angluin et al.~\cite{AAE2008} designed a monte carlo majority protocol that relies on the existence of a unique leader in $V$, but has a low probability of failing. The complexity of their protocol is, w.h.p.\footnote{In this paper, w.h.p. stands for \emph{with high probability}. An event $\mathcal E$ is said to happen  w.h.p. if $\Pr[\mathcal E] \ge 1-n^{-\Omega(1)}$.}, \textless$O(\log^2 n), O(1)$\textgreater.
	Their protocol introduces the idea of \emph{cancellations} and \emph{duplications} or \emph{splitting}. The idea behind cancellations is to introduce an empty color $e$ so that when two nodes with different non-empty colors interact, the nodes change their colors to $e$, which does not affect the majority, but the ratio between the number of majority colored nodes and the number of minority colored nodes increases.
	The intuitive idea behind splitting is that when two nodes interact where one has a non-empty color and the other has an empty color, then the color from the colored node is duplicated to also be the color of the non-empty node. Since interactions are chosen uniformly at random, intuitively, the duplications do not change the ratio between the number of majority colored nodes and the number of minority colored nodes, but the difference between these numbers increase over time.

	Bilke et al.~\cite{BCER2017} showed how to implement the protocol of Angluin et al.~\cite{AAE2008} to work without the existence of a unique leader. The complexity of their protocol is \textless$O(\log^2 n), O(\log^2 n)$\textgreater.
	Since our protocol uses ideas from the protocol of Bilke et al.~\cite{BCER2017}, we provide details of the protocol in Section~\ref{section:theclassicprotocol}.
	Alistarh et al.~\cite{AAG18} improved the complexity to \textless$O(\log^2 n), O(\log n)$\textgreater.
	The state-of-the-art population protocol for the majority problem is the \textless$O(\log^{5/3} n), O(\log n)$\textgreater\space protocol of Berenbrink et al.~\cite{BEFKKR2018}, which builds upon the protocol of Bilke et al.~\cite{BCER2017}; see Section~\ref{section:five-thirds}.

	For lower bounds, Alistarh et al.~\cite{AAG18} proved that under the conditions of monotonicity and output dominance, any majority protocol that stabilizes in expected $n^{1-\Omega(1)}$ parallel time must use at least $\Omega(\log n)$ states.
	
	\paragraph{Our results and techniques.}
	In this paper, we describe a new \textless$O(\log^{3/2} n), O(\log n)$\textgreater\space protocol for majority, thereby improving the parallel stabilization time of the state-of-the-art protocol of Berenbrink et al.~\cite{BEFKKR2018}. 
	Notice that, by the lower bound of Alistarh et al.~\cite{AAG18}, for our runtime, the number of states is optimal.
	
	Our protocol is based on the protocols of Berenbrink et al.~\cite{BEFKKR2018} and Bilke et al.~\cite{BCER2017}. In particular, the time cost of the protocol of Berenbrink et al.~\cite{BEFKKR2018} is $O(\log^{2-a}n + \log^{1+a}n)$, which is minimized when $a=1/2$.
	However,  in order to guarantee the high probability bounds, Berenbrink et al.~\cite{BEFKKR2018} require $a\le 1/3$, and so the runtime becomes $O(\log^{2-a}n + \log^{1+a}n) = O(\log^{5/3} n)$.
	
	In order to describe our protocol, we first describe in Section~\ref{section:theclassicprotocol} the protocol of Bilke et al~\cite{BCER2017}, which is the basis of the protocol of Berenbrink et al.~\cite{BEFKKR2018}.
	Then, in Section~\ref{section:five-thirds} we describe the protocol of Berenbrink et al.~\cite{BEFKKR2018}, and in particular we detail the reason for why Berenbrink et al.~\cite{BEFKKR2018} require $a\le 1/3$.
	Intuitively, the main challenge is due to maintaining counters in nodes that are incremented based on their interactions, and since the interactions are random, it is likely that values of counters deviate significantly from the average counter value.
	
	Our protocol is designed with the goal of removing the requirement that $a\le 1/3$.
	In particular, we employ a power-of-two choices strategy (\cite{ABK1999}) to reduce the deviation from the average counter value.
	The main challenge in implementing a power-of-two choices strategy is that we assume that interactions are symmetric, so if two nodes with the same state interact, their states after the interactions are necessarily the same.
	However, existing power-of-two choices strategies assume that ties can somehow be broken.
	In order to overcome this challenge, we partition the nodes that are responsible for counters to two equal sized sets, called $L$ and $R$, and use the $\lefty$ power-of-two choices strategy, which was inspired by V{\"{o}}cking~\cite{V1999} and analyzed by Berenbrink et al~\cite{BCSV2006}, to control the deviation. We remark that the proof of correctness of the protocol of Berenbrink et al~\cite{BEFKKR2018} becomes significantly simpler when implementing the $\lefty$ power-of-two choices strategy, as we describe in Section~\ref{sec:ourprotcol}.

	Moreover, in order to guarantee that there are enough nodes that are responsible for the counters (which the algorithm requires in order to progress fast enough), we show in Section~\ref{section:lognstates} that if our strategy for assigning nodes to $L$ and $R$ fails, then running a short version of the protocol of Bilke et al~\cite{BCER2017} suffices for solving majority.

	\subsection{Preliminaries}
	
	In order to differentiate between the protocol of Bilke et al.~\cite{BCER2017}, the protocol of Berenbrink et al.~\cite{BEFKKR2018}, and our new protocol, while avoiding clutter and redundancy, we use the following names for the three protocols, respectively: the $2$-protocol, the $5/3$-protocol, and the $3/2$-protocol.
	
	Our protocols often uses the \emph{broadcast} protocol, in which a node $u$ broadcasts some information to all of the other nodes in $V$. The complexity of a broadcast is \textless$O(\log n), O(1)$\textgreater, both w.h.p. and in expectation; see \cite{AG2015,AAE2008,BCER2017}.

	\section{An \texorpdfstring{\textless$O(\log^2 n),O(\log^2 n)$\textgreater}{} Protocol}\label{section:theclassicprotocol}
	Bilke et al. in~\cite{BCER2017} introduced an \textless$O(\log^2 n),O(\log^2 n)$\textgreater\space majority protocol.
	For our purposes, we provide an overview of the $2$-protocol in~\cite{BCER2017} while
	focusing only on the parallel time cost, and ignoring the number of states.
	Their protocol uses a notion of an \emph{empty} color, which we denote by $e$, together with a method of \emph{cancellation} and \emph{splitting}.
	An execution of the protocol has a sequence of at most $\log n + 2$ \emph{phases}, and each phase is composed of four stages: a cancellation stage, a first buffer stage, a splitting stage and second buffer stage.
	In the $2$-protocol, a node $u\in V$ stays in the first three stages of a phase for $\Theta(\log n)$ interactions and in the last stage for $O(\log n)$ (some nodes may be \emph{pulled out} of the last stage to the next phase).
	
	\paragraph{States.} Each node $u$ stores the following information:
	\begin{itemize}
		\item The color of $u$, $c_u \in \{b,w,e\}$.
		\item A counter $\alpha_u$, storing the phase number.
		\item A counter $\beta_u$, storing the interaction number in the current phase.
		\item A flag $\done_u$, indicating that $u$ is broadcasting $c_u$ as the majority color.
		\item A flag $\split_u$, indicating whether or not $u$ has already participated in a split during this phase (the split operation is defined below).
		\item A flag $\fail_u$, indicating whether or not $u$ has failed.
	\end{itemize}
	
	At the initialization, all of the flags and counters are set to $0$.
	Notice that it is straightforward to derive the stage of $u$ from $\beta_u$.
	
	\paragraph{The state transition function.} A node $u$ with $\fail_u=1$ broadcasts $\fail_u$. Otherwise, if $\done_u=1$ then $u$ broadcasts $\done_u$ and $c_u$.
	However, if during the broadcasting of $\done_u$ and $c_u$, $u$ interacts with a node $v$ where $c_v\ne c_u$, then $u$ sets $\fail_u\leftarrow 1$.
	
	Whenever a node $u$ interacts with node $v$, if $\done_u= \done_v=0$, and $\fail_u=\fail_v=0$, then the following rules define the state transition function (the rules all happen in parallel):
	\begin{itemize}
		\item The counters $\alpha_u,\beta_u$ are updated as follows: If $u$ and $v$ are not in the same stage and are also not in consecutive stages, then set $\fail_u \leftarrow 1$. Thus, for the rest, assume that $u$ and $v$ are either in the same stage or in consecutive stages.
		Increment $\beta_u$, and if $u$ reached the end of a second buffer stage then set $\beta_u\leftarrow0$ and increment $\alpha_u$. If $\alpha_u = \log n+2$ then set $\fail_u\leftarrow1$. If $u$ is in a second buffer stage and $v$ is in a cancellation stage, then set $\alpha_u \leftarrow \alpha_v$ (thereby pulling $u$ into the next phase), and $\beta_u \leftarrow 0$.
		If $u$ enters the second buffer stage:
		\begin{itemize}
			\item If $\split_u=0$, then set $\done_u \leftarrow 1$.
			\item If $\split_u=1$, then set $\split_u \leftarrow 0$.
		\end{itemize}

		\item  If $u$ and $v$ are both in the canceling stage, if either $c_u = w$ and $c_v = b$, or $c_u = b$ and $c_v = w$, then set $c_u\leftarrow e$.
		\item  If $u$ and $v$ are both in the splitting stage, $c_u=  e$, $c_v \ne e$, and $\split_v = 0$, then perform a \emph{split} operation on $v$ and therefore set $c_u \leftarrow c_v$ and  $\split_u \leftarrow 1$.
		\item  If $u$ and $v$ are both in the splitting stage, $c_u\ne e$, $c_v =e$ , and $\split_u = 0$, then perform a \emph{split} operation on $u$ by setting $\split_u \leftarrow 1$.
		
	\end{itemize}
	
	Notice that the protocol uses the $\split$ flags in order to guarantee that each node does not participate in more than 1 split operation per phase.
	
	\paragraph{Analysis.} Let $w_i, b_i$ be the number of white nodes and black nodes, respectively, entering phase $i$ and let $d_i = |w_i-b_i|$.
	Notice that $d_i$ does not change during a cancellation stage.
	
	Inspired by Bilke et al.~\cite{BCER2017}, Berenbrink et al.~\cite{BEFKKR2018} proved the following.
	\begin{lemma}\label{lem:classic-cancel}
		Consider the $i$th phase of the $2$-protocol. Then, w.h.p., if $|w_i-b_i|<n/3$, then after the cancellation stage of the phase, at least $6n/10$ of the nodes have an empty color, and in the splitting stage of the phase, for every node $u$ that enters the splitting stage with a non-empty color, $u$ participates in a split.
	\end{lemma}
	
	An immediate conclusion from Lemma~\ref{lem:classic-cancel} is that if $|w_i-b_i|<n/3$, then, w.h.p., $d_{i+1} = 2d_i$.
	Moreover, if for all $i\le j$, $|w_i-b_i|<n/3$, then, w.h.p., $d_j = 2^j\cdot d_0$, and so  after at most $\log (n/d_0) \leq \log n$ phases, the algorithm reaches a phase $k$ where $|w_k-b_k|\geq n/3$.
	The first such phase is called a \emph{critical} phase.
	Bilke et al.~\cite{BCER2017} proved the following lemma regarding critical phases.
	
	\begin{lemma}\label{lem:critical-phase}
		Once the $2$-protocol reaches a critical phase, then, w.h.p., the protocol undergoes at most  two more phases until there exists a node $u$ with $c_u \ne e$ that does not split in the splitting stage. Moreover, w.h.p., $c_u$ is the only color left in the system.
	\end{lemma}
	
	Notice that when there exists a node $u$ with $c_u \ne e$ that does not split in the splitting stage, then $\done_u$ is set to $1$ and $u$ broadcasts $c_u$, which costs another $O(\log n)$ parallel time.
	Finally, Bilke et al.~\cite{BCER2017} show that throughout the algorithm, w.h.p., there is never an interaction between two nodes that are at least two stages apart.
	Thus, the probability that a flag $\fail_u$ for some node $u$ is set to $1$ throughout the execution of the protocol is polynomially small in $n$.
	
	In order to guarantee that the protocol is always correct, the protocol executes the ambassador protocol~\cite{MNRS2014} in the background, which uses $O(1)$ states and $O(n\log n)$ parallel time (or $O(n^2\log n)$ interactions), both w.h.p. and in expectation. In order to determine the output of a node $u$, if $\fail_u$ is set to $1$, then the output is taken from the output of the ambassador protocol, and otherwise, the output is $c_u$.
	
	\paragraph{Complexity.} Each node participates in $O(\log n)$ interactions per phase, and the number of phases is $O(\log (n/d_0)) = O(\log n)$. Thus, if the algorithm does not set a $\fail$ flag to $1$, which happens w.h.p., then the time cost of the algorithm is $O(\log ^2n)$. Moreover, due to the ambassador algorithm running in the background, the algorithm always reaches a valid output configuration, even when some node $u$ sets $\fail_u\leftarrow 1$,  and the expected time cost is also $O(\log^2 n)$.
	
	\section{An \texorpdfstring{\textless$O(\log ^{5/3}n),O(\log n)$\textgreater}{} Protocol}\label{section:five-thirds}
	In the $2$-protocol, the $O(\log^2n)$ time bound is due to each node participating in $O(\log n)$ phases, and having $O(\log n)$ interactions in each phase.
	Berenbrink et al.~\cite{BEFKKR2018} designed a protocol that builds upon the $2$-protocol, and runs in $O(\log ^{5/3} n)$ time.
	To do so, they reduce the number of interactions in each phase to be $O(\log ^{1-a}n)$ for a specific constant $0<a<1$, to be chosen later.
	However, reducing the number of interactions in a phase affects the probability guarantee of Lemma~\ref{lem:classic-cancel}.
	In particular, it is no longer guaranteed that, w.h.p., after a non-critical phase is done, then all of the colored nodes that enter a splitting stage manage to participate in a split.
	Thus, Berenbrink et al.~\cite{BEFKKR2018} designed a mechanism that addresses the missing splits: the sequence of phases is partitioned into $\log^{1-a}n$ \emph{epochs} where each epoch is a sequence of $\log^an$ phases, and at the end of each epoch the protocol enters a special \emph{catch-up phase}. Intuitively, the role of the catch-up phase is to allow for colored nodes that did not split in some phase of the epoch to catch up on their missing splits.
	Finally, phases no longer have buffers,
	and the mechanism of \emph{pulling} nodes from one phase to a next is no longer used, and instead, Berenbrink et al.~\cite{BEFKKR2018} use the same mechanism for pulling nodes from a catch-up phase of one epoch to the next epoch.

	A node $u$ is said to be \emph{successful} in a splitting stage if either $u$ enters the stage without a color, or $u$ participates in a split during the splitting stage.
	Throughout an epoch, a node $u$ is said to be \emph{synchronized} after a splitting stage if so far $u$ was successful in all of the splitting stages of the epoch.
	If $u$ is not synchronized (that is, $u$ entered some splitting stage with a color and did not participate in a split), then $u$ is said to be \emph{out-of-sync}.
	Each node $u$ stores a flag $\sync_u$, initially set to $1$, indicating whether a node is synchronized or not.
	Throughout the execution of the protocol, if $\sync_u=0$, which means that $u$ is out-of-sync, then the protocol ignores the parts of the transition function defined in the $2$-protocol during interactions of $u$ with other nodes, except for updating the counters $\alpha_u$ and $\beta_u$, and participating in the broadcasts.
	Each node $u$ also stores a value $\phi_u$ so that if $u$ is out-of-sync, then $\phi_u$ is the phase number within the current epoch in which $u$ went out-of-sync.
	Thus, an out-of-sync node $u$ did not split in $\log^an-\phi_u$ splitting stages within the current epoch, and so, intuitively, unless this situation is directly addressed, at the end of the epoch there are $val(u) =  2^{\log^an-\phi_u}$ uncolored nodes that would have been colored with $c_u$ in the $2$-protocol.
	Notice, however, that if the epoch contains a phase that in the $2$-protocol would have been a critical phase, then this intuition is not necessarily true, since the $2$-protocol in this case may not even reach the end of the epoch.

	In order to catch up on coloring the $2^{\log^an-\phi_u}$ nodes with potentially missing colors,
	whenever an out-of-sync node $u$ interacts with an uncolored node $v$ and $\phi_u <\log^a n$, then $u$ and $v$ participate in a special split: Let  $x$ be $\phi_u$ prior to the interaction. Then $u$ sets $\phi_u \leftarrow x +1$, while $v$ sets $c_v\leftarrow c_u$, $\phi_v \leftarrow x+1$, and $\sync_v \leftarrow 0$. Notice that, after the split, $val(u)+val(v) = 2\cdot 2^{\log^a n - (x+1)} = 2^{\log^a n - x}$ which is exactly $val(u)$ prior to the interaction. Thus, intuitively, a split operation splits the value of u evenly between u and v.
	Once $u$ reaches the end of a catch-up phase, if $\phi_u = \log^a n$, then $u$ sets $\sync_u \leftarrow 1$, and is now synchronized again.
	The role of the catch-up phase is to guarantee that at the end of the  epoch, w.h.p., all of the out-of-sync nodes become synchronized again since they have enough interactions to continue splitting.

	Finally, when an out-of-sync node $u$ leaves a catch-up phase of the $j$th epoch while still being out-of-sync, then the following lemma relates $j$ to the critical phase.
	
	\begin{lemma}\label{lem:critical-epoch}[Rephrased from \cite{BEFKKR2018}]
		Suppose that the first time that there exists some node $u$ that finishes an epoch while being out-of-sync happens for epoch $j$. If $a\le \frac 13$, then, w.h.p., there exists a positive integer $k$ where $(j-2)\log^{a}n< k\leq j\log^{a}n$ such that the $k$th phase is a critical phase.
	\end{lemma}
	
	Thus, once an out-of-sync node $u$ leaves the $j$th epoch, node $u$ begins a special broadcast which initiates the $2$-protocol, but starting with the colors that were stored at nodes at the beginning of the $(j-1)$th epoch, which is the beginning of phase number $(j-2)\log^{a}n+1$.
	In order for this protocol to run, the nodes store their colors from the beginning of the last two epochs.
	Since $k-(j-2)\log^{a}n \le 2\log^an$, $d_k > n/3$ and, w.h.p., for all $i\le k$ we have $d_{i+1} = 2d_i$, then, w.h.p., $$n/3 < d_k = 2^{k-(j-2)\log^{a}n} d_{(j-2)\log^{a}n }  \le 2^{2\log^an}d_{(j-2)\log^{a}n },$$
	and so $d_{(j-2)\log^an} > \frac{n}{3\cdot 2^{2\log^an}}$.
	Thus, the $2$-protocol starting at phase $(j-2)\log^an$, runs for an additional $O(\log \frac{n}{d_{(j-2)\log^an}})= O(\log^an)$ phases, with a parallel time of $O(\log^{1+a}n)$. Notice that in order to reduce the number of states, when a node $u$ moves to the $2$-protocol, the protocol reuses the bits that were used for the epoch based protocol.

	For the total time complexity, there are at most $O(\log n)$ phases until a node $u$ ends an epoch while being out-of-sync, and in each such phase every node has $O(\log^{1-a} n)$ interactions, for a total of $O(\log^{2-a}n)$ time.
	In addition, there are $\log^{1-a}n$ epochs, and each epoch has a catch-up phase that runs for another $O(\log n)$ interactions, for a total of $O(\log^{2-a}n)$ time.
	Finally, executing the $2$-protocol after some node finishes an epoch while being out-of-sync costs another $O(\log^{1+a}n)$ time.
	Thus, the total time is $O(\log^{2-a}n + \log^{1+a}n)$.
	This runtime, is minimized when $a=\frac 12$.
	Unfortunately, due to Lemma~\ref{lem:critical-epoch}, the protocol requires $a\le \frac 13$, and so the runtime is minimized when $a=\frac 13$ and becomes $O(\log ^{5/3} n)$.

	\subsection{Barrier for Increasing \texorpdfstring{$a$}{}}
	Our new protocol, described in Section~\ref{sec:ourprotcol}, focuses on adapting the algorithm of Berenbrink et al.~\cite{BEFKKR2018} in a way that allows to increase $a$ to $\frac 12$.
	In order to explain the design choices behind our new protocol, we first provide an intuitive explanation as to what are the challenges in making Lemma~\ref{lem:critical-epoch} work for $a>\frac 13$.
	
	In general, there are two reasons for why there may exist a node that is supposed to participate in a split during a particular phase but that node does not manage to do so.
	The first reason is that every node has only $O(\log^{1-a} n)$ interactions in a phase (compared to the $O(\log n)$ interactions per phase in the $2$-protocol). We call this reason the \emph{short-phase} reason.
	Berenbrink et al.~\cite{BEFKKR2018} proved that, w.h.p., although the phases are shorter, a colored node $u$ during a splitting stage has a constant probability of interacting with a non-colored node.
	Thus, the probability that in a given phase a particular colored node $u$ does not split is at most $\frac 1 {2^{c\log^{1-a} n}}$ for some constant $c>1$.
	
	The second reason is that some nodes may have participated in a number of interactions that is far from the average number of interactions.
	In particular, nodes in a splitting stage whose counters are not concentrated near the average counter are likely to have a \emph{small} number of interactions with other nodes that are also in the splitting stage.
	We call this reason the \emph{deviation} reason.
	Berenbrink et al.~\cite{BEFKKR2018} prove that the probability of a counter being at distance at least $\Omega(\log^{1-a} n)$ from the average counter value is at most $\frac{1}{2^{c'\log^{1-2a}n}}$, for some constant $c'>1$.
	
	Since in the worst-case a node can become out-of-sync during the first phase of an epoch, and there are $\log^an$ phases in an epoch, an out-of-sync node $u$ in the beginning of the catch-up phase necessitates at most $2^{\log^a n}$ splits in order to guarantee that there are no out-of-sync nodes that originated from $u$ at the end of the catch-up phase.
	Moreover, since the probability of a node becoming out-of-sync during an epoch is at most $\frac{\log^an}{2^{c\log^{1-a}n}}+\frac{\log^an}{2^{c'\log^{1-2a}n}}$, then, w.h.p., by applying a Chernoff bound,  the number of out-of-sync nodes at the beginning of a catch-up phase is $\frac{n\log^an}{2^{\hat c\log^{1-a}n}}+\frac{n\log^an}{2^{\hat c'\log^{1-2a}n}}$, for some constants $\hat c, \hat c ' >1$.
	For the $j$th epoch, let $U_j$ denote the total number of splits needed to guarantee that there are no out-of-sync nodes at the end of the catch-up phase of epoch $j$. Notice that $U_j =\sum_{\text{out-of-sync }u} val(u)$.
	Thus, w.h.p., $U_j \le \frac{2^{\log^an}n\log^an}{2^{\hat c\log^{1-a}n}}+\frac{2^{\log^an}n\log^an}{2^{\hat c'\log^{1-2a}n}}$.
	Berenbrink et al.~\cite{BEFKKR2018} proved that, w.h.p., the number of uncolored nodes at the beginning of a catch-up phase is $\Omega(n)$, and so if $U_j = o(n)$ then, w.h.p., the number of splits that take place in order to complete the missing splits until the end of the catch-up phase is $U_j$.

	In order for the first term in the upper bound of $U_j$ to be $o(n)$, we set $a\le \frac 12$. However, In order for the second term in the upper bound of $U_j$ to be $o(n)$, we set $a\le \frac 13$. Thus, the second term in the upper bound on $U_j$, which is due to the second reason for having out-of-sync nodes, is the barrier for making the $5/3$-protocol to run in $O(\log^{3/2} n)$ time.
	
	\section{A \texorpdfstring{\textless$O(\log^{3/2}n),O(\log^2n)$\textgreater\space}{} Protocol}\label{sec:ourprotcol}
	\paragraph{Power of two choices.}
	In order to reduce the number of missing splits due to deviations, we change the $5/3$-protocol as follows: when two nodes interact, instead of both nodes incrementing their counters, the nodes implement a variation of the \emph{power-of-two choices} strategy which increments the counter only for the node with the smaller counter.
	However, since we assume that the interactions are symmetric, it is not clear what to do in case of a tie.
	For now, we assume that the interactions are asymmetric, and so in the case of a tie, the protocol is able to choose just one of the nodes to increment its counter.
	In Section~\ref{section:lognstates} we explain how to remove the requirement of asymmetric interactions.
	
	The motivation for using a power-of-two choices strategy is that this strategy has the property that counters do not deviate too much from the average counter. This property is summarized in the following invariant and theorem, from Berenbrink et al.~\cite{BCSV2006}, which we rephrase here to be expressed in our terms.
	
	\begin{invariant}\label{inv:concentration}
		Let $\alpha_i$ be the fraction of counters whose value is at least $i$  less than the average counter value.
		If $m$ interactions that increment counters have taken place so far, then, there is a constant $c_1$ such that for $1\leq i\leq c_1\log n$,
		$\alpha_i \leq 1.3 \cdot 2.8^{-i}$, w.h.p. Moreover, there is a constant $c_2$ such that for $i\geq c_2\log n$, $\alpha_i = 0$, w.h.p.
	\end{invariant}
	
	\begin{theorem}\label{theorem:powerof2}
		If $m$ interactions that increment counters have taken place so far, then, w.h.p., the value of the maximum counter  is $\frac{m}{n} + O(\log \log n)$.
	\end{theorem}
	
	By~Invariant~\ref{inv:concentration}, for $i=\hat c\log^{1-a}n$  the number of nodes whose counters are at least $i$ away from the average counter is, w.h.p.,  at most $\frac n {2^i} = \frac n {2^{\hat c\log^{1-a}n}}$.
	Thus, the probability of a node becoming out-of-sync during an epoch is at most $\frac{\log^an}{2^{c\log^{1-a}n}} + \frac{\log^an}{2^{\hat c\log^{1-a}n}}$, and so, w.h.p., the number of out-of-sync nodes at the beginning of a catch-up phase is $\frac{n\log^an}{2^{c''\log^{1-a}n}} + \frac{n\log^an}{2^{c'''\log^{1-a}n}}$, for some constants $c'',c''' > 1$.
	Moreover, w.h.p., now $U_j \le \frac{2^{\log^{a}n}n\log^an}{2^{c''\log^{1-a}n}} + \frac{2^{\log^{a}n}n\log^an}{2^{c'''\log^{1-a}n}}$, and in order for $U_j=o(n)$ we can choose $a \le 1/2$.

	Finally, the runtime of the algorithm is still $O(\log^{2-a}n + \log^{1+a}n)$, however, now we can choose $a=1/2$ to get a total runtime of $O(\log^{3/2} n)$.

	\subsection{Simpler Proofs}
	As a side effect of using the power-of-two choices strategy, we are able to provide simpler proofs for correctness. The simplifications are a byproduct of the maximum value of a counter being close to the average value.
	
	We define the \emph{beginning} of a stage to be the interaction at which $ n - n/2^{\Theta(\log^{1-a} n)}$ nodes have already entered the stage, and the \emph{end} of the stage is when the first node leaves the stage. Similar definitions are made for beginning and end of a phase.
	Recall that every node has  $O(\log^{1-a}n)$ interactions in each stage.
	
	\begin{lemma}\label{lemma:oursync}
		At the end of a stage, w.h.p., at least $n-n/2^{\Omega(\log^{1-a} n)}$ of the nodes have had $\Omega(\log^{1-a} n)$ interactions between the beginning and the end of the stage.
	\end{lemma}
	
	\begin{proof}
		When the first node exits the stage, by Theorem~\ref{theorem:powerof2}, w.h.p.,  the maximum counter value is at most $O(\log \log n)$ away from the average counter value, and, by Invariant~\ref{inv:concentration}, w.h.p., there are at most $1.3 \cdot 2.8^{-\Theta(\log^{1-a} n)}n \leq n/2^{\Omega(\log^{1-a} n)}$ nodes whose counter values are at least $\Omega(\log^{1-a} n)$ away from the average counter value.
		The lemma follows.
	\end{proof}
	
	at most $i\cdot\frac{ n}{2^{C\sqrt{\log n}}}$ out-of-sync nodes), w.h.p.
	
	The following two  lemmas are the basis for the inductive structure of the correctness. Lemma~\ref{lemma:ourcancel} states that, as long as the behaviour so far is as expected, then at the end of a cancellation stage there are many nodes with an empty color. Lemma~\ref{lemma:oursplit} states that, as long as the behaviour so far is as expected, then the number of nodes that become out-of-sync due to the splitting phase is small. We then conclude, in Corollary~\ref{cor:ourinductivephase} and Corollary~\ref{conclusion:catch-upphasebound}, that, as long as the behaviour so far is as expected, after each phase, the number of out-of-sync nodes is $o(n)$.
	Thus, the number of out-of-sync nodes entering a catch-up phase is $o(n)$, and in Lemma~\ref{lemma:suffdouble}, we use this assumption to show that, unless an epoch has a critical phase, there are no more out-of-sync nodes at the end of the epoch.
	
	\begin{lemma}\label{lemma:ourcancel}
		For a phase $i$, suppose that just before the phase begins, there are at most $o(n)$ out-of-sync nodes. If $i$ happens before the critical phase, then, w.h.p., after the end of the cancellation stage in phase $i$ at least $6n/10$ nodes have an empty color.
	\end{lemma}

	\begin{proof}
		By Lemma~\ref{lemma:oursync}, w.h.p. there are $\Theta(n\log^{1-a} n)$ interactions between the beginning of the cancelling stage and its end. Moreover, there are at least $n - o(n)$ nodes in the cancelling stage. Let $W$ be the set of synchronized nodes that are in the cancelling stage. Thus, $|W|=\Omega(n)$.
		
		Partition the $\Theta(n\log^{1-a} n)$ interactions of the canceling stage into $\Theta(\log^{1-a} n)$ periods, each period consists of $n$ interactions. Let $z_i$ be the number of nodes in $W$ at the beginning of the $i$'th period that have the minority color with respect to \emph{only} $W$; notice that it could be that the global minority color is different from the minority color with respect to $W$.
		Let $i^*$ denote the index of the last period of the stage.
		We will prove that $z_{i^*} < n/40$.
		We say that an interaction has been successful if either there was a cancellation or the number of nodes in $W$ before the interaction that have the minority color with respect to \emph{only} $W$ is (already) at most $41z_i/42$.
		
		Suppose the number of nodes in $W$ before an interaction that have the minority color with respect to \emph{only} $W$ is more than  $41z_i/42$. Then the probability that the interaction is a cancellation is at least $2\cdot(41z_i/42)^2/n^2$, and so if the number of nodes in $W$ at the end of a period that have the minority color with respect to \emph{only} $W$ is more than  $n/40$, then the expected number of cancellations within the period is $2\cdot (41z_i/42)^2/n \geq 2z_i/42$. Let $X$ be the random variable counting the number of cancellations in such a period.
		By a Chernoff bound,
		\[Pr[|X -2z_i/42| > \frac{1}{2} \cdot 2z_i/42] < \exp{(\frac{(1/2)^2\cdot 2z_i/42}{4})} =\]
		 \[\exp{(-z_i/336)}<\exp{(-n/13440)}.\]
		Therefore, in each such period, w.h.p., at least $z_i/42$ nodes from the minority colors are being canceled, and so $z_{i+1}\leq 41z_i/42$. After a sufficiently large constant number of periods, there are at most $n/40$ nodes from $W$ of the minority color. By the assumption of the lemma, the number of nodes with the majority color is at most $n/3+n/40$. The difference between the number of nodes with the majority color and the number of nodes with the minority color
		might be higher because of the out-of-sync nodes, but by no more than $o(n)$ by assumption.
		Thus, the number of nodes that are still colored after the last period is at most $n/40 +n/3+n/40 + o(n) < 4n/10$.
	\end{proof}
	
	\begin{lemma}\label{lemma:oursplit}
		Suppose that at the beginning of a phase, there are at least $\frac{6n}{10}$ uncolored nodes, and assume that the phase is before the critical phase. Then, w.h.p., at most $\frac{ n}{2^{\Omega(\log^{1-a} n)}}$ nodes become  out-of-sync during this phase.
	\end{lemma}
	
	\begin{proof}
		Recall that the two reasons for a synchronized colored node $u$ not participating in a split during the splitting stage are the short-phase reason and the deviation reason.
		Lemma~\ref{lemma:oursync} states that at most $n/2^{\Omega(\log^{1-a} n)}$ nodes did not split due to the deviation reason. Here we focus on the short-phase reason.
		
		By the definition of a beginning of a stage, there are at most $o(n)$ nodes that did not enter the stage when the stage begins. Moreover, by assumption, there are at least $\frac{6n}{10}$ nodes with an empty color when the stage begins.
		Thus, during a stage, the probability for a colored node $u$ to participate in a split during any interaction of $u$ in the stage is at least $2/10 - o(1) \ge 2/11$.
		For a node $u$ that had at least $\Omega(\log^{1-a} n)$ interactions in the splitting stage, the probability that $u$ does not participate in a split in any of the interactions is at most $2/11^{\Omega(\log^{1-a} n)} = 1/2^{\Omega(\log^{1-a} n)}$. So, the expected number of nodes that did not split is at most $n/2^{\Omega(\log^{1-a} n)}$. By a Chernoff bound, w.h.p., the number of nodes that did not split is at most $n/2^{\Omega(\log^{1-a} n)}$.

	\end{proof}

	\begin{corollary}\label{cor:ourinductivephase}
		Suppose that before a phase that appears before a critical phase, there are at most $(i-1)\frac{ n}{2^{\Omega(\log^{1-a} n)}}$ nodes that are out-of-sync.
		Then, w.h.p., after the phase, at most $i\frac{ n}{2^{\Omega(\log^{1-a} n)}}$ nodes are out-of-sync.
	\end{corollary}
	
	\begin{corollary}\label{conclusion:catch-upphasebound}
		
		Suppose that epoch $j$ happens before the critical phase. Then $U_j = o(n)$.
	\end{corollary}

	\begin{lemma}\label{lemma:suffdouble}
		Consider an epoch $j$ that appears before a critical phase.
		If $U_j = o(n)$, then, w.h.p., at the end of the catch-up phase of epoch $j$ all of the nodes are synchronized.
	\end{lemma}
	
	\begin{proof}
		Recall that $U_j =\sum_{\text{out-of-sync }u} val(u)$.
		We partition the missing splits that are counted by $U_j$ to singleton splits.
		When an interaction takes place between an out-of-sync node $u$ and an non-colored node $v$, $u$ passes $val(u)/2$ of its singletons to $v$. Thus, each singleton in $u$ needs to interact with $\log val(u) \le \log^a n$ non-colored nodes before the end of the catch-up phase in order to guarantee that after the catch-up phase there are no more out-of-sync nodes.
		
		By Lemma~\ref{lemma:ourcancel}, there are at most $\frac{4n}{10}$ synchronized colored nodes after the last cancellation stage in the epoch, and each one of these nodes could be split at most once during the last splitting stage. Thus, there are at most $\frac{8n}{10}$ synchronized colored nodes at the beginning of the catch-up phase. In addition, there are at most $o(n)$ out-of-sync nodes at the beginning of the catch-up phase, and so the number of synchronized nodes with an empty color at the beginning of the catch-up phase is more than $\frac{n}{10}$.

		By Invariant~\ref{inv:concentration} and Theorem~\ref{theorem:powerof2}, and similar to the proof of Lemma~\ref{lemma:oursync}, w.h.p., there are $\Omega(n\log n)$ interactions during a catch-up phase, and so, w.h.p., every singleton interacts with at least $\Omega(\log^a n)$ uncolored nodes. Thus, w.h.p., there are no out-of-sync nodes by the end of the catch-up phase.
	\end{proof}

	\begin{lemma}
		If an epoch contains a critical phase, then after either this epoch or the next epoch, there is at least one out-of-sync node.
	\end{lemma}
	
	\begin{proof}
		Assume by contradiction that the lemma is false. This means that in both epochs all of the colored nodes managed to split. Since there are at least $n/3$ more nodes with the majority color than nodes with the minority color at the beginning of the last phase of the first epoch, then right before the second epoch, the difference is at least $2n/3$. However, this means that after the second epoch there are at least $4n/3$ colored nodes, which is a contradiction.
	\end{proof}
	
	\section{Improvement to an \texorpdfstring{\textless$O(\log^{3/2}n),O(\log n)$\textgreater\space}{} Protocol}\label{section:lognstates}
	In this section, we describe how to reduce the number of states to $O(\log n)$.
	The main idea, inspired by Alistarh et al.~\cite{AAG18} and Berenbrink et al~\cite{BEFKKR2018}, is to partition the nodes to two types of nodes: \emph{workers} and \emph{clocks}.
	Moreover, the protocols no longer use the counters $\alpha_u$ and $\beta_u$ from the $2$-protocol. Instead, we introduce a new counter $\gamma_u$ that counts until $O(\log n)$.
	The counter $\gamma_u$ is used differently, depending on whether $u$ is a clock or a worker.

	Intuitively, a worker node $u$ runs the protocol from Section~\ref{sec:ourprotcol}, while counting the phase number using $\gamma_u$.
	We also add another $O(1)$ bits per node to indicate the current stage.
	Notice that the number of the current epoch can be derived from the phase counters.

	Clock nodes do not participate in the protocol from Section~\ref{sec:ourprotcol}.
	Instead, they are responsible for the information needed to move a worker node from one phase to the next, and from one stage to the next.
	Intuitively, clock nodes keep track of progress within an epoch.
	Each  clock node $u$ uses $\gamma_u$ as a counter that is incremented with applying a variant of the power-of-two choices strategy whenever two clock nodes interact; the details of this strategy are given below.
	When a clock increments its counter  and the counter reaches $\Theta(\log n)$, the clock sets the counter back to 0, indicating that a new epoch has begun.

	A worker node $u$ moves from one stage to the next stage when $u$ interacts with a clock whose counter indicates that the phase has progressed to the next stage.
	Similarly, $u$ moves from one phase to the next phase when $u$ interacts with a clock whose counter indicates that the phase has progressed to the next phase.
	Notice that the same ideas allow for moving into and out of a catch-up phase.
	
	\paragraph{Left-bias power-of-two choices.} The power-of-two choices strategy that we apply is the $\lefty$ variation: the clock nodes are partitioned into two equal-sized sets, the \emph{Left} set and the \emph{Right} set, and counters are incremented only during interactions of two clock nodes from different sets.
	In such a case, when the counters of the two clock nodes are equal, the clock from the Left set increments its counter. Otherwise, the clock with the smaller counter increments its counter. We show in Lemma~\ref{lemma:invl12} that Invariant~\ref{inv:concentration} holds for $\lefty$. Notice that \cite{BCSV2006} proved that Theorem~\ref{theorem:powerof2} holds for $\lefty$.
	
	\paragraph{Creating workers and clocks.}
	We now describe the procedure for creating workers and clocks.
	The very first interaction of a node $u$ is called the \emph{initial interaction of $u$}.
	When two nodes $u$ and $v$ have an interaction that is an initial interaction for both nodes, where $u$ is a black node and $v$ is a white node, then $u$ becomes a Right-clock and $v$ becomes a Left-clock. such an interaction is called a \emph{first cancellation}.
	Any other type of initial interaction makes the node a worker node.
	After the initial interactions take place, it is straightforward to see that the number of Right-clocks is exactly the number of Left-clocks.
	Moreover, the majority color within the working nodes is the same as the majority color for all of the nodes.
	
	In order to show that our protocol works, w.h.p., we  show that the number of clocks and the number of workers is $\Omega(n)$. However, if the number of nodes with the minority color is $o(n)$, and so $d_0$ (which is the difference between the number of nodes with the majority color and the number of nodes with the minority color) is very large, then it is impossible to guarantee that there are $\Omega(n)$ clocks.
	However, in such a case, if the workers execute the $2$-protocol (Section~\ref{section:theclassicprotocol}) then by Lemma~\ref{lem:critical-phase}, w.h.p., the workers find the majority color within two phases of the $2$-protocol.
	
	Thus, a worker node $u$ right after the initial interaction simulates two phases of  the $2$-protocol, while using the bits of $\gamma_u$ to implement $\alpha_u$ and $\beta_u$; notice that the bits in $\gamma_u$ suffice, since there are only $O(\log n)$ interactions during the first two phases of the $2$-protocol.
	During these two phases, interactions between $u$ and clock nodes are ignored.
	If $u$ does not split during one of the phases, then $u$ initiates a broadcast protocol with the color $c_u$ as the majority.
	Otherwise, $u$ reverts to the $3/2$-protocol, but now using the clock nodes in order to keep track of phases.
	Notice that when $u$ moves from the two phases of  the $2$-protocol to the new protocol, the information that was in $\alpha_u$ is overwritten, since now $\gamma_u$ is used for the new protocol.

	At the same time, a clock node $v$ counts until some value $t = \Theta(\log n)$, using the power-of-two choices strategy, where $t$ is chosen to be large enough so that, w.h.p., all of the worker nodes have completed their two phases of the $2$-protocol. After the value $t$ is reached, the next increment on the counter in $v$ resets the counter to $0$.
	Notice that a worker node $u$ that finished the two phases of the $2$-protocol begins executing the protocol from Section~\ref{sec:ourprotcol} after interacting with a clock node $v$ whose counter has been reset to $0$ at least once.
	
	\paragraph{Correctness.} If $d_0>n/3$, then there will be $\Omega(n)$ workers, and so, w.h.p.,  by Lemma~\ref{lem:critical-phase}, after at most two more phases of the $2$-protocol, there exists a node broadcasting the majority color.
	Thus, for the following assume that $d_0 \le n/3$, and so there are more than $n/3$ nodes of the minority color. The following lemma shows that we have $\Theta(n)$ clocks and $\Theta(n)$ workers.
	
	\begin{lemma}\label{lemma:clocksworkersthetan}
		If $d_0 \le n/3$ then,  w.h.p., the protocol produces $\Theta(n)$ clocks and $\Theta(n)$ workers.
	\end{lemma}
	
	\begin{proof}
		Since $d_0 \le n/3$, there are at least $n/3$ nodes with the minority color.
		In each of the first $n/12$ interactions, there is a constant probability of at least $1/24$ to have a first cancellation. 
		Thus, there are at least $n/24$ clocks in expectation after $n/12$ interactions, and by a Chernoff bound, w.h.p., at least $n/30$ clocks.
		
		In each of the first $n/12$ interactions, there is a constant probability of at least $1/24$ to \emph{not} have a first cancellation. Similarly, w.h.p., there are at least $n/30$ worker nodes.
	\end{proof}

	Let $n_w$ be the number of worker nodes and $n_c$ be the number of clock nodes.
	Next we show that most of the worker nodes are highly synchronized in their progress.
	Notice that Lemma~\ref{lemma:oursync} still holds, as now there are less increments of counters, and so, w.h.p., almost all the nodes have at least $\Omega(\log ^{1-a}n)$ interactions between the beginning and the end of each phase.
	In the following lemma we show that most of the workers have enough interactions \emph{in} every phase.
	
	\begin{lemma}\label{lemma:workersync}
		At the end of the $i$th phase, w.h.p., all but at most $i\cdot n_w/2^{\Theta(\log ^{1-a}n)}$ of the worker nodes have $\Theta(\log ^{1-a}n)$ interactions in phase $i$ since the beginning of the phase.
	\end{lemma}
	
	\begin{proof}
		Let $t_s, t_e$  be the largest counter value of a clock node when phases $i$ and $i+1$ begin, respectfully, and let $t = t_e - t_s = \Theta(\log^{1-a}n)$.
		We say that a phase \emph{is in time} $\tau$ at the first time a clock node in the phase has a counter value of $t_s + \tau$.
		By Invariant~\ref{inv:concentration} and Theorem~\ref{theorem:powerof2}, when the phase is in  time $t/3$, w.h.p., all but at most $n_c/2^{\Omega(\log ^{1-a}n)}$ of the clock nodes are in the phase.

		Now consider the moment when the phase is in time $2t/3$.
		Let $W_{k}$ be the set of nodes that are updated to be in phase $k$. By induction, prior to the $i$th phase, there were $|W_{i-1}| = n_w(1-(i-1)\cdot n_w/2^{\Theta(\log ^{1-a}n)})$ worker nodes that were in phase $i-1$. By Lemma~\ref{lemma:oursync}, all but $n_w/2^{\Omega(\log ^{1-a}n)}$ of the worker nodes have $\Theta(\log ^{1-a}n)$ interactions between time $t/3$ and $2t/3$.
		Since there are $\Theta(n)$ clock nodes in the system and w.h.p., at time $t/3$ most of clock nodes are in the phase, then the probability that an interaction between time $t/3$ and time $2t/3$ involving a worker node in phase $i-1$ is with a clock node in phase $i$ is constant.
		Thus, the probability that a worker node is not in phase $i$ by time $2t/3$ is at most $1/2^{\Omega(\log ^{1-a}n)}$.
		Therefore, the expected number of worker nodes that are not in phase $i$ by time $2t/3$ is at most $|W_{i-1}|/2^{\Omega(\log ^{1-a}n)}$, and so by a Chernoff bound, w.h.p., the number of worker nodes that are not in phase $i$ by time $2t/3$ is at most $|W_{i-1}|/2^{\Omega(\log ^{1-a}n)}$.
		
		Finally, by Lemma~\ref{lemma:oursync}, between time $2t/3$ and time $t$, all but $n_w/2^{\Omega(\log ^{1-a}n)}$ of the worker nodes have had $\Theta(\log ^{1-a}n)$ interactions.
		
		In conclusion, the number of worker nodes that did not have $\Omega(\log ^{1-a}n)$ interactions in phase $i$ is, w.h.p., at most $n_w/2^{\Omega(\log ^{1-a}n)}$.
	\end{proof}
	
	By Lemma~\ref{lemma:workersync} and Lemma~\ref{lemma:oursplit}, in each epoch that is before a critical phase, w.h.p., there are at most $o(n)$ worker nodes that are out-of-sync entering the catch-up phase of the epoch, as desired by Lemma~\ref{lemma:suffdouble}.
	Thus, we conclude the following:
	
	\begin{lemma}\label{lemma:ourrunningtime}
		After $O(\log^{1.5} n)$ parallel time, w.h.p., the protocol reaches a stable correct configuration.
	\end{lemma}

	In order to achieve the desired running time in expectation, our protocol runs the ambassador protocol in the background. If our protocol succeeds, then the output is determined by our protocol, and otherwise, the output is determined by the ambassador protocol, which is a low probability event. Thus, the expected time is still $O(\log^{1.5} n)$.
	
	\begin{theorem}
		The exact majority can be computed in $O(\log^{3/2} n)$ parallel time and $O(\log n)$ states per node, w.h.p. and in expectation.
	\end{theorem}
	
	\section{Symmetric interactions}
	
	In this section we show that Invariant~\ref{inv:concentration} holds for algorithm $\lefty$. We prove the invariant for a corresponding version of $m$ balls and $n$ bins, following the steps of \cite{BCSV2006}.
	
	Define a \emph{batch} as $n$ consecutive balls, and define \emph{time} $t$ to be the time after $t$ batches.
	Notice that $t$ is the average number of balls per bin at time $t$, i.e., after allocating $tn$ balls.
	The \emph{height} of a ball $i$ in a bin $u$ is the number of balls in $u$ after inserting ball $i$.
	Let $l_t^u$ be the number of balls in bin $u$ at time $t$.
	
	\begin{lemma}\label{lemma:lowbinsnextint}
		Let $a_l\cdot n$ and $b_l\cdot n$ be an upper bound on the number of bins with at most $l$ balls on the left side and the right side, respectively. The probability that a bin $u$ on the \emph{right} side with exactly $l$ balls will receive a ball in the next interaction is at least $\frac{2-2a_l}{n}$. The probability that a bin $v$ on the \emph{left} side with exactly $l$ balls will receive a ball in the next interaction is at least $\frac{2-2b_l}{n}$.
	\end{lemma}
	
	\begin{proof}
		Let $u$ be a bin with exactly $l$ balls on the right side.
		In order for $u$ to receive a ball, $u$ must interact with a bin on the left side that has more than $l$ balls. There are at least $(1-a_l)n$ bins on the left side with more than $l$ balls, and there are $(b_l-b_{l-1})n$ bins on the right side with exactly $l$ balls. The probability that in the next interaction, a bin on the right side with exactly $l$ balls will receive a ball is at least $2(b_l-b_{l-1})(1-a_l)$. Each of the bins with exactly $l$ balls are equally likely to receive a ball, so the probability that a specific bin $u$ receives a ball is at least $2(b_l-b_{l-1})(1-a_l)/(b_l-b_{l-1})n = \frac{2-2a_l}{n}$. Similarly, the probability that a node $v$ of the left side with exactly $l$ balls receives a ball is at least $\frac{2-2b_l}{n}$ (the probability is actually slightly larger).
	\end{proof}
	
	Assuming Invariant~\ref{inv:concentration} holds on time $t-1$, and applying Invariant~\ref{inv:concentration} on Lemma~\ref{lemma:lowbinsnextint}, we have the following.
	
	\begin{corollary}\label{observation:2/n}
		The probability of a bin $u$ with at most $t-5$ balls after batch $t-1$ to receive a ball in the next allocation is at least $1.9/n$.
	\end{corollary}
	The following lemma corresponds to Invariant~\ref{inv:concentration}:
	
	\begin{lemma}\label{lemma:invl12}
		Let $t\geq 0$. Assume that there is no batch $\tau<t$ such that Invariant~\ref{inv:concentration} failed in the beginning of batch $\tau$. Then, w.h.p., Invariant~\ref{inv:concentration} holds after batch $t$.
	\end{lemma}
	
	\begin{proof}
		Consider a bin $u$. If $l_t^u < t$, then we say that $u$ has $t-l_t^u$ \emph{holes}.
		Let $q_t = t - l_t^u$. Assume $q_t = 5 + i$. Let $t'<t$ be the last batch such that $q_{t'} = 5$. Notice that the number of balls that $u$ receives throughout batches $t'$ to $t$ is at most $t-t'+i$. Moreover, by Lemma~\ref{lemma:lowbinsnextint} and since we assume Invariant~\ref{inv:concentration} holds until time $t-1$, for each ball $i$ in these batches, the probability that $i$ is allocated to $u$ is at least $1.9/n$.
		The number of balls that are allocated into bin $u$ is asymptotically dominated by a binomial random variable $B((t-t')n,1.9/n)$.
		We obtain:
		
		\begin{align*}
		\Pr[q_t\geq i+5 ]
		&\leq \sum_{t'=0}^{t-1}\Pr [B((t-t')n,1.9/n)\leq t-t'-i]  \\
		&\leq  \sum_{\tau = 1}^{\inf}\sum_{k=i}^{\tau}  \Pr [B((t-t')n,1.9/n)\leq t-t'-i].
		\end{align*}
		
		Berenbrink et al.~\cite{BCSV2006} showed that for every $0\leq k \leq \tau$,
		\[
		\Pr [B((t-t')n,1.9/n)\leq t-t'-i] \leq 3.4^{-k}\cdot3.4^{-0.1\tau}.
		\]
		Hence, $\Pr[q_t\geq i+5 ] \leq 13.5\cdot 3.4^{-i}$.
		
		Let $Q_t$ denote the number of holes in the bin with the least number of balls at time $t$.
		Thus,
		\[
		\Pr[Q_t\geq i+5]\leq n\cdot13.5\cdot3.4^{-i}.
		\]
		Thus, $Q_t=O(\log n)$ w.h.p., and we have the last argument of Invariant~\ref{inv:concentration}.
		For the first part, we have:
		\[
		\Pr[q_t\geq i]\leq \cdot13.5\cdot3.4^{-i+5} \leq 0.65\cdot2.8^{-i},
		\]
		where the second inequality holds for $i>41$. By the union bound, there are at most $1+0.65\cdot2.8^{-i}$ bins with at most $t-i$ balls. By Chernoff, w.h.p., there are at most $1.3\cdot 2.8^{-i}+2$ bins with at most $t-i$ balls.
	\end{proof}

	\bibliographystyle{plainurl}
	\bibliography{main_bibl}

\begin{thebibliography}{10}

\bibitem{AAEGR2017}
Dan Alistarh, James Aspnes, David Eisenstat, Rati Gelashvili, and Ronald~L.
  Rivest.
\newblock Time-space trade-offs in population protocols.
\newblock In {\em Proceedings of the Twenty-Eighth Annual {ACM-SIAM} Symposium
  on Discrete Algorithms, {SODA} 2017, Barcelona, Spain, Hotel Porta Fira,
  January 16-19}, pages 2560--2579, 2017.
\newblock \href {https://doi.org/10.1137/1.9781611974782.169}
  {\path{doi:10.1137/1.9781611974782.169}}.

\bibitem{AAG18}
Dan Alistarh, James Aspnes, and Rati Gelashvili.
\newblock Space-optimal majority in population protocols.
\newblock In {\em Proceedings of the Twenty-Ninth Annual {ACM-SIAM} Symposium
  on Discrete Algorithms, {SODA} 2018, New Orleans, LA, USA, January 7-10,
  2018}, pages 2221--2239, 2018.
\newblock \href {https://doi.org/10.1137/1.9781611975031.144}
  {\path{doi:10.1137/1.9781611975031.144}}.

\bibitem{DBLP:conf/dna/AlistarhDKSU17}
Dan Alistarh, Bartlomiej Dudek, Adrian Kosowski, David Soloveichik, and
  Przemyslaw Uznanski.
\newblock Robust detection in leak-prone population protocols.
\newblock In Robert Brijder and Lulu Qian, editors, {\em {DNA} Computing and
  Molecular Programming - 23rd International Conference, {DNA} 23, Austin, TX,
  USA, September 24-28, 2017, Proceedings}, volume 10467 of {\em Lecture Notes
  in Computer Science}, pages 155--171. Springer, 2017.

\bibitem{AG2015}
Dan Alistarh and Rati Gelashvili.
\newblock Polylogarithmic-time leader election in population protocols.
\newblock In {\em Automata, Languages, and Programming - 42nd International
  Colloquium, {ICALP} 2015, Kyoto, Japan, July 6-10, 2015, Proceedings, Part
  {II}}, pages 479--491, 2015.
\newblock \href {https://doi.org/10.1007/978-3-662-47666-6\_38}
  {\path{doi:10.1007/978-3-662-47666-6\_38}}.

\bibitem{AGV2015}
Dan Alistarh, Rati Gelashvili, and Milan Vojnovic.
\newblock Fast and exact majority in population protocols.
\newblock In {\em Proceedings of the 2015 {ACM} Symposium on Principles of
  Distributed Computing, {PODC} 2015, Donostia-San Sebasti{\'{a}}n, Spain, July
  21 - 23, 2015}, pages 47--56, 2015.
\newblock \href {https://doi.org/10.1145/2767386.2767429}
  {\path{doi:10.1145/2767386.2767429}}.

\bibitem{AADF2004}
Dana Angluin, James Aspnes, Zo{\"{e}} Diamadi, Michael~J. Fischer, and
  Ren{\'{e}} Peralta.
\newblock Computation in networks of passively mobile finite-state sensors.
\newblock In {\em Proceedings of the Twenty-Third Annual {ACM} Symposium on
  Principles of Distributed Computing, {PODC} 2004, St. John's, Newfoundland,
  Canada, July 25-28, 2004}, pages 290--299, 2004.
\newblock URL: \url{https://doi.org/10.1145/1011767.1011810}.

\bibitem{AAE2008}
Dana Angluin, James Aspnes, and David Eisenstat.
\newblock Fast computation by population protocols with a leader.
\newblock {\em Distributed Computing}, 21(3):183--199, 2008.
\newblock \href {https://doi.org/10.1007/s00446-008-0067-z}
  {\path{doi:10.1007/s00446-008-0067-z}}.

\bibitem{AR2009}
James Aspnes and Eric Ruppert.
\newblock An introduction to population protocols.
\newblock In {\em Middleware for Network Eccentric and Mobile Applications},
  pages 97--120. 2009.
\newblock \href {https://doi.org/10.1007/978-3-540-89707-1\_5}
  {\path{doi:10.1007/978-3-540-89707-1\_5}}.

\bibitem{ABK1999}
Yossi Azar, Andrei~Z. Broder, Anna~R. Karlin, and Eli Upfal.
\newblock Balanced allocations.
\newblock {\em {SIAM} J. Comput.}, 29(1):180--200, 1999.
\newblock \href {https://doi.org/10.1137/S0097539795288490}
  {\path{doi:10.1137/S0097539795288490}}.

\bibitem{BCSV2006}
Petra Berenbrink, Artur Czumaj, Angelika Steger, and Berthold V{\"{o}}cking.
\newblock Balanced allocations: The heavily loaded case.
\newblock {\em {SIAM} J. Comput.}, 35(6):1350--1385, 2006.
\newblock \href {https://doi.org/10.1137/S009753970444435X}
  {\path{doi:10.1137/S009753970444435X}}.

\bibitem{BEFKKR2018}
Petra Berenbrink, Robert Els{\"{a}}sser, Tom Friedetzky, Dominik Kaaser, Peter
  Kling, and Tomasz Radzik.
\newblock A population protocol for exact majority with o(log5/3 n)
  stabilization time and theta(log n) states.
\newblock In {\em 32nd International Symposium on Distributed Computing, {DISC}
  2018, New Orleans, LA, USA, October 15-19, 2018}, pages 10:1--10:18, 2018.
\newblock \href {https://doi.org/10.4230/LIPIcs.DISC.2018.10}
  {\path{doi:10.4230/LIPIcs.DISC.2018.10}}.

\bibitem{BCER2017}
Andreas Bilke, Colin Cooper, Robert Els{\"{a}}sser, and Tomasz Radzik.
\newblock Brief announcement: Population protocols for leader election and
  exact majority with \emph{O}(log\({}^{\mbox{2}}\) \emph{n}) states and
  \emph{O}(log\({}^{\mbox{2}}\)\emph{ n}) convergence time.
\newblock In {\em Proceedings of the {ACM} Symposium on Principles of
  Distributed Computing, {PODC} 2017, Washington, DC, USA, July 25-27, 2017},
  pages 451--453, 2017.
\newblock \href {https://doi.org/10.1145/3087801.3087858}
  {\path{doi:10.1145/3087801.3087858}}.

\bibitem{DBLP:conf/icalp/CzyzowiczGKKSU15}
Jurek Czyzowicz, Leszek Gasieniec, Adrian Kosowski, Evangelos Kranakis, Paul~G.
  Spirakis, and Przemyslaw Uznanski.
\newblock On convergence and threshold properties of discrete lotka-volterra
  population protocols.
\newblock In Magn{\'{u}}s~M. Halld{\'{o}}rsson, Kazuo Iwama, Naoki Kobayashi,
  and Bettina Speckmann, editors, {\em Automata, Languages, and Programming -
  42nd International Colloquium, {ICALP} 2015, Kyoto, Japan, July 6-10, 2015,
  Proceedings, Part {I}}, volume 9134 of {\em Lecture Notes in Computer
  Science}, pages 393--405. Springer, 2015.

\bibitem{DV2010}
Moez Draief and Milan Vojnovic.
\newblock Convergence speed of binary interval consensus.
\newblock In {\em {INFOCOM} 2010. 29th {IEEE} International Conference on
  Computer Communications, Joint Conference of the {IEEE} Computer and
  Communications Societies, 15-19 March 2010, San Diego, CA, {USA}}, pages
  1792--1800, 2010.
\newblock \href {https://doi.org/10.1109/INFCOM.2010.5461999}
  {\path{doi:10.1109/INFCOM.2010.5461999}}.

\bibitem{DBLP:conf/spaa/GasieniecSU19}
Leszek Gasieniec, Grzegorz Stachowiak, and Przemyslaw Uznanski.
\newblock Almost logarithmic-time space optimal leader election in population
  protocols.
\newblock In Christian Scheideler and Petra Berenbrink, editors, {\em The 31st
  {ACM} on Symposium on Parallelism in Algorithms and Architectures, {SPAA}
  2019, Phoenix, AZ, USA, June 22-24, 2019}, pages 93--102. {ACM}, 2019.

\bibitem{KU2018}
Adrian Kosowski and Przemyslaw Uznanski.
\newblock Brief announcement: Population protocols are fast.
\newblock In {\em Proceedings of the 2018 {ACM} Symposium on Principles of
  Distributed Computing, {PODC} 2018, Egham, United Kingdom, July 23-27, 2018},
  pages 475--477, 2018.
\newblock URL: \url{https://dl.acm.org/citation.cfm?id=3212788}.

\bibitem{MNRS2014}
George~B. Mertzios, Sotiris~E. Nikoletseas, Christoforos~L. Raptopoulos, and
  Paul~G. Spirakis.
\newblock Determining majority in networks with local interactions and very
  small local memory.
\newblock In {\em Automata, Languages, and Programming - 41st International
  Colloquium, {ICALP} 2014, Copenhagen, Denmark, July 8-11, 2014, Proceedings,
  Part {I}}, pages 871--882, 2014.
\newblock \href {https://doi.org/10.1007/978-3-662-43948-7\_72}
  {\path{doi:10.1007/978-3-662-43948-7\_72}}.

\bibitem{V1999}
Berthold V{\"{o}}cking.
\newblock How asymmetry helps load balancing.
\newblock In {\em 40th Annual Symposium on Foundations of Computer Science,
  {FOCS} '99, 17-18 October, 1999, New York, NY, {USA}}, pages 131--141, 1999.
\newblock \href {https://doi.org/10.1109/SFFCS.1999.814585}
  {\path{doi:10.1109/SFFCS.1999.814585}}.

\end{thebibliography}

\end{document}